\documentclass[nofootinbib,superscriptaddress,twocolumn,aps,pra]{revtex4}

\usepackage{amsmath,amsfonts,amssymb}
\usepackage{wrapfig}
\usepackage{graphicx}
\usepackage{bbm}
\usepackage{color}
\usepackage{float}
\usepackage{subfigure}
\usepackage{textcomp}
\usepackage[english]{babel}
\usepackage{pifont}
\usepackage{thmtools}

\declaretheorem{theorem}
\declaretheorem{lemma}
\usepackage{algorithm}

\newenvironment{proof}[1][Proof]{\begin{trivlist}
\item[\hskip \labelsep {\bfseries #1}]}{\end{trivlist}}

\newenvironment{corollary}[1][Corollary]{\begin{trivlist}
\item[\hskip \labelsep {\bfseries #1}]}{\end{trivlist}}
\newenvironment{lemma1again}[1][Lemma 1]{\begin{trivlist}
\item[\hskip \labelsep {\bfseries #1}]}{\end{trivlist}}
\newenvironment{lemma2again}[1][Lemma 2]{\begin{trivlist}
\item[\hskip \labelsep {\bfseries #1}]}{\end{trivlist}}

\newcommand{\qed}{\nobreak \ifvmode \relax \else
      \ifdim\lastskip<1.5em \hskip-\lastskip
      \hskip1.5em plus0em minus0.5em \fi \nobreak
      \vrule height0.75em width0.5em depth0.25em\fi}

\usepackage{mathrsfs}
\usepackage{wasysym} % for circled letters
\usepackage{psfrag}
\usepackage[colorlinks=true,linkcolor=blue,citecolor=blue]{hyperref}
\usepackage{enumerate} % to call \begin{enumerate}[(i)]

\newcommand{\be}{\begin{equation}}
\newcommand{\ee}{\end{equation}}
\newcommand{\bea}{\begin{eqnarray}}
\newcommand{\eea}{\end{eqnarray}}
\newcommand{\ba}{\begin{array}}
\newcommand{\ea}{\end{array}}
\newcommand{\bc}{\begin{center}}
\newcommand{\ec}{\end{center}}
\newcommand{\ben}{\begin{enumerate}}
\newcommand{\een}{\end{enumerate}}
\newcommand{\bi}{\begin{itemize}}
\newcommand{\ei}{\end{itemize}}
\newcommand{\bt}{\begin{table}}
\newcommand{\et}{\end{table}}
\newcommand{\btab}{\begin{tabular}}
\newcommand{\etab}{\end{tabular}}
\newcommand{\bfi}{\begin{figure}}
\newcommand{\efi}{\end{figure}}

\newcommand{\bd}{\begin{description}}
\newcommand{\ed}{\end{description}}

\newcommand{\nn}{\nonumber}

\newcommand{\e}{\mathrm e}

\makeatletter

\def\compoundrel#1\over#2{\mathpalette\compoundreL{{#1}\over{#2}}}
\def\compoundreL#1#2{\compoundREL#1#2}
\def\compoundREL#1#2\over#3{\mathrel
      {\vcenter{\hbox{$\m@th\buildrel{#1#2}\over{#1#3}$}}}}
\makeatother

\newcommand{\djj}{d\kern-0.4em\char"16\kern-0.1em}

\usepackage{babel}
\usepackage{wasysym}
\usepackage{harmony}
\usepackage{semtrans}

\makeatother

\newcommand{\ket}[1]{\left|#1\right\rangle }

\newcommand{\pmset}{\{\ket{\pm}\}}

\def\be{\begin{eqnarray}}
\def\ee{\end{eqnarray}}
\def\bee{\begin{eqnarray*}}
\def\eee{\end{eqnarray*}}

\begin{document}

\title{Entanglement of $\pi$-LME states and the SAT problem}

\author{Adi Makmal}
\affiliation{Institut f{\"u}r Theoretische Physik,
Universit{\"a}t Innsbruck, Technikerstra{\ss }e 25, A-6020 Innsbruck}
\affiliation{Institut f{\"u}r Quantenoptik und Quanteninformation der
\"Osterreichischen Akademie der Wissenschaften, Innsbruck, Austria}
\author{Markus Tiersch}
\affiliation{Institut f{\"u}r Theoretische Physik,
Universit{\"a}t Innsbruck, Technikerstra{\ss }e 25, A-6020 Innsbruck}
\affiliation{Institut f{\"u}r Quantenoptik und Quanteninformation der
\"Osterreichischen Akademie der Wissenschaften, Innsbruck, Austria}
\author{Vedran Dunjko}
\affiliation{Institut f{\"u}r Theoretische Physik,
Universit{\"a}t Innsbruck, Technikerstra{\ss }e 25, A-6020 Innsbruck}
\affiliation{Institut f{\"u}r Quantenoptik und Quanteninformation der
\"Osterreichischen Akademie der Wissenschaften, Innsbruck, Austria}
\affiliation{Laboratory of Evolutionary Genetics, Division of Molecular Biology, Ru\djj er Bo\v{s}kovi\'{c} Institute, Bijeni\v{c}ka cesta 54, 10000 Zagreb, Croatia.} 
\author{Shengjun Wu}
\affiliation{Kuang Yaming Honors School, Nanjing University, Nanjing, Jiangsu 210093, P. R. China}

\date{\today}

\begin{abstract}
In this paper we investigate the entanglement properties of the class of $\pi$-locally maximally entanglable ($\pi$-LME) states, which are also known as the \emph{real equally weighted states} or the \emph{hypergraph states}. The $\pi$-LME states comprise well-studied classes of quantum states (e.g.\ graph states) and exhibit a large degree of symmetry. Motivated by the structure of LME states, we show that the capacity to (efficiently) determine if a $\pi$-LME state is entangled would imply an efficient solution to the boolean satisfiability (SAT) problem. 
More concretely, we show that this particular problem of entanglement detection, phrased as a decision problem, is $\mathsf{NP}$-complete. The restricted setting we consider yields a technically uninvolved proof, and illustrates that entanglement detection, even when quantum states under consideration are highly restricted, still remains difficult.
\end{abstract}

\maketitle

%%%%%%%%%%%%%%%%%%%%%%%%%%%%%%%%%%%%%%%%%%%%%%%%%%%%%%%%%%%%%%%%%%%%%%%%%%%%%%%%%%%%%%%%%%%%%%%%%%%%%%%%%%%%%%%%%%%%%%%%%%%%%

\section{Introduction} \label{sec:intro}
Determining whether a given quantum state is entangled has been studied during the past decade in many different contexts \cite{Horodecki_RMP_2009}. In some settings the state is given as a physical system, in which case one encounters a detection problem, whereas in others it is described classically, e.g.\ in a form of a density matrix, where one is then to solve a decision problem. Both detection \cite{Terhal_2002,Guehne_2009} and decision \cite{Gurvits_2003,2007_Ioannou_QIC,Gharibian_2010,2013_Milner_IEEE,2013_Milner_Arx} problems of entanglement have been considered, and the computational complexity of the latter problems have been analyzed. 

In either context, checking if a general ($n$-partite) quantum state is entangled is a hard problem even when it is pure. By definition, for a pure $n$-partite state one would have to show that the state is not fully separable, i.e.\ $\ket{\psi}\neq\ket{\psi_1}\otimes\ldots\otimes\ket{\psi_n}$, where $\ket{\psi_j}$ describes the $j^{th}$ subsystem. For example, consider a quantum state of $n$ qubits $\ket{\psi}=\alpha_0\ket{0\dotsc 00}+\alpha_1\ket{0\dotsc 01}+\dotsb+\alpha_{2^n-1}\ket{1\dotsc1}$, where all the coefficients are complex and are only restricted by the normalization of $\ket{\psi}$. Changing just a single coefficient may turn a fully separable state to an entangled one and vice versa, so in general we need to take into account all $2^n$ coefficients in order to answer whether $\ket{\psi}$ is entangled.

In this paper we consider a seemingly easier problem by restricting to a class of pure quantum states with high symmetry, and show that, nonetheless, determining whether these states are entangled or not is difficult. 
The class of states, which we consider for the rest of the paper, are $\pi$-locally maximally entanglable ($\pi$-LME) states \cite{Carle_2013} (also known as \emph{real equally weighted} states \cite{2012_Qu_Arx}) in which all of the coefficients $\alpha_j$ are either $+1$ or $-1$:
\begin{equation}
	\label{eq:signLME}
	\ket{\psi_f} = \frac{1}{\sqrt{2^n}} \sum_{\vec{s}\in\{0,1\}^{n}} (-1)^{f(\vec{s})} \ket{\vec{s}},
\end{equation}
where $f(\vec{s})$ is a Boolean function of the $n$-bit string, $\vec{s}$. With this family of states we are still confronted with exponentially many coefficients, but only allow for a binary choice for their values rather than a continuum. 
Consequently, while for general pure $n$-partite states the set of fully separable states is of measure zero, in the class of $\pi$-LME states it is of finite volume.

The set of $\pi$-LME states is a subset of the standard LME-states~\cite{LMEstates}, an interesting set of states, which are defined, up to local unitaries, as 
$\frac{1}{\sqrt{2^n}} \sum_{\vec{s}\in\{0,1\}^{n}}\e^{i \varphi(\vec{s})} \ket{\vec{s}}$.
That is, when writing them in the computational basis, 
all basis states appear with the same probability but generally different phases $\varphi(\vec{s})$. The standard LME states include important sets of states such as all stabilizer states or weighted graph states \cite{LMEstates}, which are useful for quantum error correction schemes \cite{Gottesmann_1997}. In particular, by setting $\varphi(\vec{s})=\frac{\pi}{2} \vec{s}\cdot\!\Gamma\!\cdot\vec{s}$, with $\Gamma$ being an adjacency matrix of a graph of $n$ nodes, one captures all graph states \cite{Raussendorf_2003,Hein_2006}, which are a resource for measurement based quantum computation \cite{Raussendorf_2001}. 
In our case of $\pi$-LME states $\ket{\psi_f}$, the phases are restricted by $\varphi(\vec{s})=\pi f(\vec{s})$. The $\pi$-LME states thus include all graph states and were recently shown to equal the more extended set of \emph{hypergraph states} \cite{2013_Rossi,2013_Qu_PRA}.

Physically, hypergraph states may naturally arise in the model of spin gases \cite{Calsamiglia_2005}. In such a system, a classical gas of particles, each of which carries a qubit degree of freedom that is initially set to $\ket{+}$ (where we use the convention $\ket{\pm}=\frac{\ket{0}\pm\ket{1}}{\sqrt{2}}$), undergoes dynamics during which particles collide. These collisions effectively induce a controlled phase gate onto the participating qubits. 
For example, a two-particle collision induces a 2-qubit phase gate, which corresponds to drawing an edge between the two relevant nodes in the corresponding graph. Allowing for analogous $k$-body interactions gives rise to $k$-qubit phase gates \cite{LMEstates} and to the notion of an edge involving $k$ nodes, i.e., a hyperedge in the hypergraph \cite{2013_Qu_PRA,2013_Rossi}. 
The corresponding Boolean function $f(\vec{s})$ can then be written as a polynomial in $\vec{s}$, involving terms of degree $k$.

A hypergraph state, or equivalently, a $\pi$-LME can therefore be efficiently generated using phase gates involving up to all $n$ qubits, acting on the state $\ket{+}^{\otimes n}$ \cite{LMEstates}. Accordingly, if those phase gates, or, alternatively, the adjacency tensor of the underlying hypergraph, are known then answering if the resulting state $\ket{\psi_f}$ is entangled is, in fact,   straightforward: the state is fully separable if and only if no entangling gate is present, i.e.\ the underlying hypergraph is fully disconnected (there exist no edges) \cite{LMEstates}. To avoid confusion, we stress that in the present work the function $f(\vec{s})$ is given as a Boolean formula and we assume no knowledge on its polynomial form. We will, however, come back to this point later on.

To tackle the question of separability of the $\pi$-LME states we first note that a $\pi$-LME state, $\ket{\psi_f}$, is fully separable if and only if it is a product of the $\ket{\pm}$ states up to a global sign, i.e.\ belongs to the orthonormal set $\pm\{\ket{\pm}\}^{\otimes n}$ \cite{2012_Qu_Arx}.  
This observation simplifies the problem a little: to decide if $\ket{\psi_f}$ is fully separable it is sufficient to check if it is a product of $\ket{\pm}$ states.
We denote the latter property as ``orthonormal set membership (OSM) of $\ket{\psi_f}$ with respect to the $\{\ket{\pm}\}^{\otimes n}$ basis set", and abbreviate it formally as OSM($\ket{\psi_f}, \{\ket{\pm}\}^{\otimes n}$), where from now on, to ease notation, we replace the specification of this particular type of state and basis with an asterisk, and simply write $OSM^*$. Accordingly, performing $OSM^*$ would mean ``determining if $\ket{\psi_f} \in \pm\{\ket{\pm}\}^{\otimes n}$" and we say that a state $\ket{\psi_f}$ is an $OSM^*$ state if and only if $\ket{\psi_f} \in \pm\{\ket{\pm}\}^{\otimes n}$.

We then continue and show that using quantum computation, the ability to perform $OSM^*$ (given as a black box) is powerful enough to solve the satisfiability problem (SAT) efficiently.
Next we phrase the problem as a decision problem in the language of computational complexity theory and show that a similar result holds also for classical computation. 
Moreover, we show that deciding whether given a Boolean function $f$ the corresponding $\pi$-LME state is entangled, or equivalently, not an $OSM^{*}$ state, is an $\mathsf{NP}$-complete problem. This provides an additional link between the difficulty of experimental entanglement detection and the hardness of the related computational problem.

The paper is structured as follows: 
First, Section \ref{sec:complete_separability} states several observations on $\pi$-LME states, which connect between the $OSM^*$ properties of these states, their entanglement properties, and the structure of their generating Boolean function $f$. 
Next, following a few prerequisites in Section \ref{secsec:prerequisites}, we use those observations to show in section \ref{secsec:reduction_proof} that  performing $OSM^*$ would allow for an efficient solution of the SAT problem.
The question of $OSM^*$ is then rephrased as a standard decision problem in Section \ref{secsec:classical_complexity} and its complexity is analyzed.
Taking a different perspective, section \ref{sec:comparing_to_other_tasks} draws relations between the $OSM^*$ task and the task of perfect state discrimination, which sheds light on the (im)possibility of realizing perfect or approximated $OSM^*$ operations.  
Finally, Section \ref{sec:conclusions} connects our $OSM^*$ observations back to the separability question of $\pi$-LME states, and concludes the paper.

\section{Structural observations on $\pi$-LME states} \label{sec:complete_separability}
In what follows we repeat and extend several basic mathematical observations on $\pi$-LME states which were shown in \cite{2012_Qu_Arx}. 
For completeness, we provide (independent) proofs in Appendix \ref{app:second_lemma}, but note that the corollary of the first two lemmas was previously proved in \cite{2012_Qu_Arx}. 

We start with noting that for a $\pi$-LME state $\ket{\psi_f}$ the question of separability is equivalent to asking if the state is an $OSM^*$ state, i.e.\ belongs to the set $\{\ket{\pm}\}^{\otimes n}$:
\begin{lemma} \label{lemma:1}
A $\pi$-LME state $\ket{\psi_f}$ (as given in Eq.\ (\ref{eq:signLME})), is fully separable if and only if each of its $n$ qubits is in either state $\ket{+}$ or in state $\ket{-}$.
\end{lemma}
\begin{proof} See Appendix \ref{app:second_lemma}. 
\end{proof} 

Lemma \ref{lemma:1} implies that in order to decide if $\ket{\psi_f}$ is fully separable it is sufficient to check if it is a product of $\ket{\pm}$ states. Moreover, we can now explicitly see that out of $2^{2^n}$ possible $\pi$-LME states (of $n$ qubits), exactly $2^{n+1}$ states are products of the $\pmset$ states (including a factor of 2 for the global sign) and are therefore fully separable, whereas the rest must be entangled.
Note that for general pure states  $\ket{\psi}$, Lemma \ref{lemma:1} does not hold: clearly, if $\ket{\psi}$ is not fully separable then it cannot be a product of $\pmset$ states, but the inverse is not true \footnote{In particular,  a single particle state can be in either $\ket{\pm}$ state, whereas entanglement is not even defined in this case.}. For example, the product state $\ket{\Psi}=\ket{00}=\ket{++}+\ket{+-}+\ket{-+}+\ket{--}$ is not a product of the $\pmset$ states.

We now turn to the second lemma, which connects properties of the Boolean function $f$ to the $OSM^*$ structure of the $\pi$-LME states:
\begin{lemma} \label{lemma:2}
A $\pi$-LME state $\ket{\psi_f}$ is a product of $\ket{\pm}$ states (up to a global sign) only if the Boolean function $f(\vec{s})$ is either balanced ($|\{\vec{s}, f(\vec{s})=0\}|=|\{\vec{s}, f(\vec{s})=1\}|=2^{n-1}$) or constant ($f(\vec{s})=c \; \forall \vec{s}$).
\end{lemma}
\begin{proof}
See Appendix \ref{app:second_lemma}. 
\end{proof}
Combining Lemma \ref{lemma:1} and Lemma \ref{lemma:2} we get that a $\pi$-LME state $\ket{\psi_f}$ is fully separable only if the Boolean function $f$ is either balanced or constant, as was also shown in \cite{2012_Qu_Arx}. This relates the question of separability of $\pi$-LME states to a necessary property of the corresponding Boolean function $f$.

Note that the other direction of Lemma \ref{lemma:2} does not hold: there exist balanced-sign states $\ket{\psi_f}$ (i.e.\ of balanced function $f$),  that are not 
product states in the $\ket{\pm}$ basis, e.g.\ the 3-qubit state
\be
\label{eq:balanced_but_not_product}
\ket{\psi} = &\frac{1}{\sqrt{8}}&(-\ket{000}+\ket{001}+\ket{010}+\ket{011} \nn \\
 &{}&\;-\ket{100}+\ket{101}-\ket{110}-\ket{111})\nn \\
=  &\frac{1}{\sqrt{2}}&(-\ket{+0-}+\ket{-1+}),
\ee
which is the GHZ state \cite{GHZ_state_1989} up to local unitary transformations (constant-sign states, of constant function $f$, are always given, however, by the product states $\pm\ket{+,\ldots,+}$). 
The same can be shown by the following counting argument (a similar analysis was carried in \cite{2011_Bruss,2012_Qu_Arx}):
For $n$ bits, there are $N\!=\!2^n$ possible strings and the number of balanced-sign states is given by the number of possible ways to choose $N/2$ strings (of minus sign) out of total N, i.e.\ $N \choose N/2$. On the other hand, the available number of product basis states $\{\ket{\pm}\}^{\otimes n}$ (including a global sign) is only $2N\!\!=\!\!2^{n+1}$, including the two constant states. 
Since ${N \choose N/2} > 2N-2$ for all $N\!=\!2^n$ with $n\geq 3$, it follows that some balanced-sign states are not products of the $\ket{\pm}$ states (for $n=1,2$ all balanced-sign states are products of $\ket{\pm}$ states). 

Finally, to identify which of the balanced-sign states are products of the $\ket{\pm}$ states, we take a closer look at the tensor product structure in the computational basis $\{\ket{0},\ket{1}\}$
\be
\label{eq:tensor_product}
\ket{\psi}\!=\!\ket{\pm}_{n-1}\!\otimes\!\ldots\!\otimes\ket{\pm}_0\!=\!\frac{1}{\sqrt{2^n}}
\!\left(\!\!
\begin{array}{c}
+1\\
\pm 1
\end{array}
\!\!\right)_{\!\!n-1}\!\!\!\!\!\!\!\!
\otimes\!\ldots\!\otimes
\!\left(\!\!
\begin{array}{c}
+ 1\\
\pm 1
\end{array}
\!\!\right)_{\!\!0}\!\!.
\ee
Then, when expanding $\ket{\psi}$ in the computational basis $\{\ket{0},\ket{1}\}^{\otimes n}$ it is seen that the $k^{th}$ qubit (from right to left, $0\leq k \leq n-1$) is in state $\ket{+}$ ($\ket{-}$) if and only if the first $2^k$ terms (from $\ket{0}$ to $\ket{2^k-1}$) have exactly the same (opposite) signs as the next $2^k$ terms (from $\ket{2^k}$ to $\ket{2^{k+1}-1}$). This proves the following lemma:
\begin{lemma}
\label{lemma:3}
A state $\ket{\psi}$ is a product of the $\pmset$ states if and only if for all $0\leq k \leq n-1$, the first $2^k$ basis states have the same or opposite signs as the next $2^k$ basis states, when $\ket{\psi}$ is expanded in the computational basis $\{\ket{0},\ket{1}\}^{\otimes n}$.
\end{lemma}
Lemma \ref{lemma:3} shows once again that the state of Eq.\ (\ref{eq:balanced_but_not_product}) is not $OSM^*$ since the signs of, e.g., the first two basis states $(-,+)$ are neither the same as, nor exactly opposite from the signs of the next two basis states (+,+). Note also that Lemma \ref{lemma:3} also implies Lemma \ref{lemma:2}.

In what comes next we link these structural observations on $\pi$-LME states to the structure of the SAT problem.

\section{Solving SAT by performing $OSM^*$} \label{sec:reduction}
In this section we show that a quantum computer (QC) aided with  an oracle which determines whether $\ket{\psi_f}$ is an $OSM^*$ state, can solve the SAT problem. 

\subsection{Prerequisites} \label{secsec:prerequisites}
For the benefit of the reader, we first introduce the basic notation. 
First, the satisfiability problem (SAT) is defined as follows (see e.g.\ \cite{Garey_1979}): Let $f:\{0,1\}^{n} \!\rightarrow\! \{0,1\}$ be a Boolean function, then SAT outputs ``yes" if there exist some assignments for $\{x_i\}_{i=1}^n$ such that $f(x_1,\ldots,x_n)=1$, and ``no" otherwise.

Next, to evaluate the Boolean function $f$, the algorithm we present uses  
a quantum circuit that implements the unitary transformation $U_f$ (for any classical circuit that implements $f$, there exists a quantum circuit of comparable efficiency \cite{NielsenChuang}), defined as: 
\be \label{eq:oracle_rule}
U_f: \ket{x}\ket{y} \rightarrow  \ket{x}\ket{y \oplus f(x)},
\ee
where $x$ is a $n$-bit string, $y$ is a single bit, 
and $\oplus$ is addition modulo $2$. Recall that \cite{NielsenChuang}:
\be
U_f: \ket{+,\ldots,+}\ket{-} \rightarrow 
\frac{1}{\sqrt{2^n}}\sum_{\vec{s}\in\{0,1\}^{n}} (-1)^{f(\vec{s})} \ket{\vec{s}}\ket{-}, 
\ee
such that the ancillary qubit $\ket{-}$ is decoupled, and we are left with a
$\pi$-LME state $\ket{\psi_{f}}$, as in Eq.\ (\ref{eq:signLME}).
The $\pi$-LME states thus arise naturally in several quantum algorithms (see also \cite{2011_Bruss}) and in particular in the Deutsch-Jozsa algorithm \cite{DeutschJozsa_1992}, as we next use.

\subsection{Solving SAT on a QC with an $OSM^*$ black box} \label{secsec:reduction_proof}
Given an efficiently computable Boolean function $f:\{0,1\}^{n} \rightarrow \{0,1\}$, Alg.\ \ref{prot:SAT} determines if $f$ has a satisfying assignment (SAT) or not by performing $OSM^*$ (for proof of correctness see Appendix \ref{app:proof}). 
Here, we assume a quantum computer empowered with an $OSM^*$ black box which, given a $\pi$-LME state $\ket{\psi_{f}}$ as an output of a quantum circuit, reports if $\ket{\psi_{f}}$ is a product of the $\ket{\pm}$ states. For efficiency considerations we note that throughout Alg.\ \ref{prot:SAT}, $OSM^*$ is called only once, and the function $f$ is evaluated (via $U_f$) at most four times (or three, had $OSM^*$ been a non-demolition operation).

\begin{algorithm}[htbp]
\caption[The protocol]{Solving SAT with $OSM^*$} \label{prot:SAT}
\begin{enumerate}
 \setlength{\itemsep}{-1pt}
 \item Given the Boolean function $f$, construct the quantum circuit that implements $U_f$ as defined in Eq.\ (\ref{eq:oracle_rule}) and apply $U_f$ on $\ket{+,\ldots,+}\ket{-}$ to obtain the state $\ket{\psi_f}\ket{-}$. 
 \item Perform $OSM^*$ on $\ket{\psi_{f}}$. If the answer is ``no" then OUTPUT  \textbf{``satisfying assignment exists"} and BREAK; Otherwise, i.e.\ if the answer is ``yes", continue.
\item Apply the Deutsch-Jozsa algorithm \cite{DeutschJozsa_1992}. If $f$ is balanced then OUTPUT: \textbf{``satisfying assignment exists"} and BREAK; Otherwise, i.e.\ if $f$ is constant, continue.
\item Apply $U_f$ on e.g.\ $\ket{0,\ldots,0}\ket{0}$ and measure the ancilla in the $\{\ket{0},\ket{1}\}$ basis. If $\ket{0}$ is measured then OUTPUT: \textbf{``satisfying assignment does not exist"} and BREAK; Otherwise, i.e.\ if $\ket{1}$ is measured then OUTPUT: \textbf{``satisfying assignment exists"}.
\vspace{-7pt}
\end{enumerate}
\end{algorithm}

This algorithm connects the SAT problem with the ability to perform $OSM^*$, which is by Lemma \ref{lemma:1} equivalent to detecting entanglement for $\pi$-LME states. Since Lemma \ref{lemma:1} actually shows that a $\pi$-LME state is entangled if and only if it is \emph{not} an $OSM^*$ state, it is more natural to consider the task of not-OSM membership, which we denote as 
$\lnot OSM^{*}$. This is not of practical importance in the operational sense, but will become relevant later on in the coming section. 
Informally, the relations between the different tasks can therefore can be summarized as
\be
\label{eq:tasks_relation}
ENT\stackrel{\text{def}}{\geq_{QC}} ENT^{*}\stackrel{\text{lemma}  \ref{lemma:1}}{=}\lnot OSM^{*} \stackrel{\text{Alg. \ref{prot:SAT}}}{\geq_{QC}} SAT, 
\ee
where ENT is entanglement detection for any pure state, ENT$^{*}$ is entanglement detection for $\pi$-LME states (given as an output of a quantum circuit), and the $A \geq_{QC} B$ relation means that a QC with access to a black box A can solve B efficiently.

In this section we assumed, for simplicity, that the $OSM{^*}$ black-box device  distinguishes perfectly between $OSM^*$ and non-$OSM^*$ $\pi$-LME states in one shot. This capacity, however, is clearly impossible within quantum mechanics. 
We could relax our restrictions and merely consider imperfect OSM devices, e.g., of polynomially bounded away errors ($1/2+\mathcal{O}(1/poly(n))$), which would also allow an efficient solution, via standard amplification methods. However, as we clarify in Section \ref{sec:comparing_to_other_tasks}, all such devices, which could be used to solve SAT efficiently, are forbidden by quantum mechanics.
Nonetheless, the description of a Boolean function $f$ is classical, and the problem whether such a function generates an $OSM^*$ or non-$OSM^*$ state  $\ket{\psi_f}$, is computable. In the next section we thus study the complexity of this problem, based on Alg.\ \ref{prot:SAT}, which establishes a simple relationship between the SAT problem and the computational analogous problem of the $OSM^*$ task, or equivalently, between the SAT problem and the separability problem of $\pi$-LME states.

\subsection{$OSM^*$ as a decision problem} 
\label{secsec:classical_complexity}
The question of separability, as a computational problem, was formulated in several different ways \cite{2007_Ioannou_QIC}. 
In particular, deciding if a given density matrix is separable was shown to be $\mathsf{NP}$-hard with respect to Turing reductions, when allowing for an error that scales as an inverse polynomial in the density matrix size \cite{Gurvits_2003,Gharibian_2010}. In addition, when specifying the input in terms of the quantum circuit which outputs a certain quantum state (either pure or mixed), several formal connections to complexity classes have been established in Refs.\ \cite{2013_Milner_IEEE,2013_Milner_Arx}. 
In what follows we tackle the question of separability of $\pi$-LME states by studying  the equivalent problem of $OSM^*$, rephrased as a decision problem. 

A Turing reduction can be thought of as ``an algorithm
that solves one problem by using as a subroutine an algorithm
that solves another problem`` \cite{Oded_Goldreich}. 
Here we show that SAT can be Turing reduced to $OSM^*$ on a classical deterministic Turing machine, when the later is recast as a decision problem: 
\be
\label{eq:turing_red}
 cOSM^* \geq_T^P SAT,
\ee
 where $\geq_T^P$ stands for polynomial-time Turing reductions and where $cOSM^*$, which stands for computational $OSM^*$, is defined as the following decision problem: given a classical description of a Boolean function $f$, decide if $\ket{\psi_f} \in \pm\{\ket{\pm}\}^{\otimes n}$. 

To show this formally we provide a polynomial time Turing reduction from SAT to $cOSM^*$, that is we show that there is a classical algorithm that solves SAT efficiently with a polynomial number of calls (in the number of qubits $n$) to $cOSM^*$. To that end we follow most of the steps of Alg.\ \ref{prot:SAT} with little modification: step (1) is unnecessary since $cOSM^*$ receives $f$ as an input, and step (4), which only evaluates the function $f$ on some input, can easily be performed in polynomial time on a deterministic Turing machine.

The only step in Alg.\ \ref{prot:SAT} which has to be done differently on a classical computer is step (3) in which the Deutsch-Jozsa algorithm is used to differentiate between balanced and constant functions $f$. 
Here, instead, we construct a new Boolean function $g = f(x_1,\ldots,x_n) \land x_{n+1}$. It turns out, as shown in Appendix \ref{app:classical_reduction}, that the function $g$ is constant or balanced if and only if the function $f$ is constant. Hence, by applying  $cOSM^*$ on $g$ we can distinguish between constant and balanced function $f$, as required.

The SAT problem is known to be $\mathsf{NP}$-complete. Eq.\ (\ref{eq:turing_red}) therefore implies that under Turing reduction, $cOSM^*$ is hard for $\mathsf{NP}$ (and trivially so for co-$\mathsf{NP}$, since $\mathsf{NP}$-hard = co-$\mathsf{NP}$-hard with respect to the Turing reduction).
In Appendix \ref{app:classical_reduction} we show an even stronger result: that there exists a Karp reduction from SAT to $\lnot cOSM^*$
\be
\label{eq:Karp_reduction}
	\lnot cOSM^* \geq_K^P SAT,
\ee
where $\geq_K^P$ stands for Karp (polynomial-time many-one) reduction. This implies that $\lnot cOSM^*$ is a Karp $\mathsf{NP}$-hard problem.

We now turn to show that the problem of $\lnot cOSM^*$ is in $\mathsf{NP}$. To that end we use Lemma \ref{lemma:3} to show that verifying that a state $\ket{\psi_f}$ is not an OSM$^*$ state can be done quickly: 
The certificate will be composed of three integer numbers $0\leq k \leq n-1, 0\leq l,m \leq 2^k-1$ with which one could verify that the signs coefficients of the terms $\ket{l}$ and $\ket{m}$ are neither exactly the same nor exactly opposite from those of the terms $\ket{2^{k}+l}$ and $\ket{2^{k}+m}$, when expanding $\ket{\psi_f}$ in the computational basis $\{\ket{0},\ket{1}\}^{\otimes n}$. In other words, the values of $f(l)$ and $f(m)$ are neither exactly the same nor exactly opposite from those of $f(2^{k}+l)$ and $f(2^{k}+m)$. This can be verified with only four evaluations of $f$. For example, in the case of the non-$OSM^*$ state of Eq.\ (\ref{eq:balanced_but_not_product}), a possible certificate would be $k=1,l=0,m=1$, which focuses on the the first two terms ($\ket{000}, \ket{001}$) and the following two ($\ket{010}, \ket{011}$). 
The decision problem $\lnot cOSM^*$ is thus in $\mathsf{NP}$ and is therefore shown to be an $\mathsf{NP}$-complete problem.

Combining the reduction of Eq.\ (\ref{eq:Karp_reduction}) with Lemma $\ref{lemma:1}$ we immediately get that the decision problem cENT$^{*}$, defined as ``given a classical description of a Boolean function $f$, decide if the state $\ket{\psi_f}$ is entangled" is also (Karp) $\mathsf{NP}$-complete.

Before moving on, we take an additional look at $\ket{\psi_f}$ as a hypergraph state. As mentioned in the introduction, if the underlying hypergraph of $\ket{\psi_f}$ is known then deciding if it is fully separable is easy: $\ket{\psi_f}$ is fully separable if and only if its underlying hypergraph has no edges. Our result therefore implies that finding a polynomial time algorithm which converts a Boolean function $f$ to the corresponding hypergraph would entail that $\mathsf{NP}=\mathsf{P}.$

\section{$OSM^*$ and state discrimination} \label{sec:comparing_to_other_tasks}
In the simplest version of state discrimination (SD), a system is being prepared in one of two pure states $\ket{\psi_1}$ or $\ket{\psi_2}$, with equal probabilities. The task is then to decide which of the two states describes the system. Quantum mechanics does not allow perfect SD between non-orthogonal states (see, e.g.\ \cite{NielsenChuang}). In fact, perfect SD would violate the no-cloning \cite{2000_Chefles} and the no-signaling principles \cite{Barnett_2009_review}, and was shown, from the point of view of complexity theory, to suffice for solving the Unique-SAT problem \cite{Abrams_PRL_1998,Aaronson_2005} (i.e.\ SAT with the promise of having at most a single satisfying assignment \cite{Valiant–Vazirani_86}). 
Moreover, for exponentially close states, the Helstrom's bound \cite{Helstrom} provides an optimal error probability which is exponentially close to $\frac{1}{2}$ (implying a complete lack of knowledge).

Indeed, the $OSM^*$ ability entail in certain instances distinguishing between exponentially close states. In particular, perfect SD between the $\pi$-LME states $\ket{\psi_{f_1}}=\ket{+,\ldots,+}$ and $\ket{\psi_{f_2}}=\ket{+,\ldots,+} - \frac{2}{\sqrt{2^n}}\ket{0,\ldots,0}$ (corresponding to no, or a unique satisfying assignment of the characteristic functions $f_1$ and $f_2$, respectively), can be reduced to $OSM^*$: for $\ket{\psi_{f_1}}$, $OSM^*$ answers positively, whereas for $\ket{\psi_{f_2}}$ it answers negatively. 
Performing perfect $OSM^*$ is thus forbidden within quantum mechanics. Note further that $\ket{\psi_{f_1}}$ and $\ket{\psi_{f_2}}$ are not only non-orthogonal, but are also exponentially close (with $n$, the number of qubits). The Helstrom's bound then limits the success probability of imperfect $OSM^*$ in such a way that even a polynomial bounded away error (from $\frac{1}{2}$) would require an exponential number of copies of the two states. This means, in particular, that any attempt to solve SAT efficiently via an  approximated determination of $OSM^*$, is doomed to fail.

\section{Conclusion} \label{sec:conclusions}
We showed that the capacity to perform $OSM^*$, i.e.\ to determine if a $\pi$-LME state is a product of the $\ket{\pm}$ states up to a global sign, would allow for an efficient solution of the SAT problem. Such an ability is, however, forbidden by quantum mechanics. 
In fact, even imperfect $OSM^*$, which is however ``useful", in the sense that it could allow an efficient SAT solution, via standard amplification methods, cannot be realized.

Motivated by this observed relation between $OSM^*$ and SAT, we have defined the analogous decision problem, $cOSM^*$, and showed that it is $\mathsf{NP}$-complete. 
Showing hardness was based on the above connection to SAT, whereas the $\mathsf{NP}$-membership was crucially dependent on the tensor product structure of a multi-partite quantum state. Overall we see that the difficulty of performing $OSM^*$ matches the hardness of the analogous $cOSM^*$ decision problem. This is perhaps unsurprising as it is in concord with the common belief that a quantum computer cannot solve $\mathsf{NP}$-complete problems efficiently, i.e.\ that $\mathsf{BQP}$ does not contain $\mathsf{NP}$ (see, e.g.\ \cite{Aaronson_2005}).

Due to the equivalency between separability and OSM-membership on the set of $\pi$-LME states, shown at the very beginning of this paper, all our observations regarding the $OSM^*$ property, hold trivially for separability as well. In particular, this analysis shows that: (a) Approximated entanglement detection with an error polynomially bounded away from $\frac{1}{2}$, would require, even in the constraint setup of $\pi$-LME states, resources that would scale exponentially with the number of qubits; (b) The decision problem cENT$^*$, i.e.\ given a Boolean function, determine if the corresponding $\pi$-LME state is entangled, is $\mathsf{NP}$-complete. 
This analysis thus provides a new simple example of the connection between concepts in physics and complexity theory. 

\begin{acknowledgments}
We acknowledge support by the Austrian Science Fund (FWF) through the SFB FoQuS: F\,4012.
\end{acknowledgments}

\appendix
\section{Proof of Lemmas  \ref{lemma:1}-\ref{lemma:2}} \label{app:second_lemma}
To ease the reading, we state the Lemmas once again:

\begin{lemma1again} 
A $\pi$-LME state, $\ket{\psi_f}$ (as given in Eq.\ (\ref{eq:signLME})), is fully separable if and only if each of its $n$ qubits is in either state $\ket{+}$ or in state $\ket{-}$.
\end{lemma1again}
\begin{proof} The ``if" side is trivial. 
We show the ``only if" side by first separating off the last qubit and writing the sum in \eqref{eq:signLME} as a sum over the basis states of the first $n-1$ qubits with bit string $\vec{r} \in \{0,1\}^{n-1}$ and the last qubit:
\begin{align}
	\ket{\psi_f} &= \frac{1}{\sqrt{2^n}} \sum_{b=0}^1\sum_{\vec{r}} (-1)^{f(b,\vec{r})} \ket{b} \ket{\vec{r}}  \nn \\
	&= \frac{1}{\sqrt{2^n}} \sum_{\vec{r}} \left[ (-1)^{f(0,\vec{r})} \ket{0} + (-1)^{f(1,\vec{r})} \ket{1} \right] \ket{\vec{r}}  \nn \\
	&= \frac{1}{\sqrt{2^n}} \sum_{\vec{r}} (-1)^{\tilde{f}(\vec{r})} \left[ \ket{0} + (-1)^{g(\vec{r})} \ket{1} \right] \ket{\vec{r}} . 
\end{align} 
By assumption, $\ket{\psi_f}$ is a product state and therefore $g(\vec{r})$ must be constant. The state of the last qubit is thus either $\ket{+}$ (for $g(\vec{r})=0$), or $\ket{-}$ (for $g(\vec{r})=1$). The remaining qubits factorize in the same way iteratively.
$\blacksquare$
\end{proof}

\begin{lemma2again}
A $\pi$-LME state, $\ket{\psi_f}$, is a product of $\ket{\pm}$ states (up to a global sign) only if the Boolean function $f(\vec{s})$ is either balanced ($|\{\vec{s}, f(\vec{s})=0\}|=|\{\vec{s}, f(\vec{s})=1\}|=2^{n-1}$) or constant ($f(\vec{s})=c \; \forall \vec{s}$).
\end{lemma2again}

\begin{proof}
We first show that if a state $\ket{\psi}$ is a product of $\ket{\pm}$ states up to a sign, then in the $\{\ket{0},\ket{1}\}^{\otimes n}$ basis, it is either that all its terms are of the same sign (``constant sign" state), or that half are positive and half are negative (``balanced sign" state). At times we just write ``a state is constant (balanced)" to simplify notation. The proof is by induction on $n$, the number of qubits: 
\bd
	\item[Basis step (n=1):] 
\ed
\begin{itemize}
\vspace{-10pt}
    \item $\ket{\psi} = \pm\ket{+} = \pm \left(\frac{\ket{0}+\ket{1}}{\sqrt{2}}\right)$, where both terms have the same sign (constant sign state).
\vspace{-7pt}
	 \item $\ket{\psi} = \pm\ket{-} = \pm \left(\frac{\ket{0}-\ket{1}}{\sqrt{2}}\right)$, where one term is negative and one is positive (balanced sign state).
\vspace{-7pt}
\end{itemize} 
For $n=1$, a $\pm\ket{\pm}$ state is either constant or balanced.

\bd
	\item[Induction step: ] We assume correctness for all $n \leq m$ and show that it also holds for $m\!\!+\!\!1$. 
Let $\ket{\psi_{m+1}}$ be a product of $m\!\!+\!\!1$ $\ket{\pm}$ states, then   $\ket{\psi_{m+1}} = \pm\ket{\pm}\ket{\psi_m}$, where $\ket{\psi_m}$ is, by assumption, either constant or balanced. The state $\ket{\psi_{m+1}}$ must then take one of the following forms:
\ed

\begin{itemize}
    \item $\ket{\psi_{m+1}} = \pm\ket{+}\ket{\psi_m} = \pm \frac{\ket{0}\ket{\psi_{m}}+\ket{1}\ket{\psi_{m}}}{\sqrt{2}} $
	 \begin{itemize}
		\item If $\ket{\psi_{m}}$ is constant then  $\ket{\psi_{m+1}}$ is constant.
		\item If $\ket{\psi_{m}}$ is balanced then  $\ket{\psi_{m+1}}$ is balanced.
	\end{itemize}

\vspace{-7pt}
	 \item $\ket{\psi_{m+1}} = \pm\ket{-}\ket{\psi_m} = \pm\frac{\ket{0}\ket{\psi_{m}}-\ket{1}\ket{\psi_{m}}}{\sqrt{2}} $
	 \begin{itemize}
		\item If $\ket{\psi_{m}}$ is constant then  $\ket{\psi_{m+1}}$ is balanced.
		\item If $\ket{\psi_{m}}$ is a balanced then  $\ket{\psi_{m+1}}$ is balanced.
	\end{itemize}
\end{itemize} 
\vspace{-7pt}
Thus $\ket{\psi_{m+1}}$ must be either a constant- or a balanced-sign state, thereby proving correctness for any $n$. 
\end{proof}
\begin{corollary}
If a $\pi$-LME state, $\ket{\psi_f}$, is a product of $\ket{\pm}$ states then it is either a constant- or balanced-sign state in the $\{\ket{0},\ket{1}\}^{\otimes n}$ basis. This implies, due to the particular structure of  $\ket{\psi_{f}}$, that the function $f$ is either constant or balanced, which concludes the proof of Lemma \ref{lemma:2}.  \quad $\blacksquare$
\end{corollary}

\section{Correctness of Alg. (1)} \label{app:proof}
\begin{theorem}
  Algorithm \ref{prot:SAT} decides SAT.  
\end{theorem}
\begin{proof}
The proof follows the steps of the algorithm:
At the end of step (1), we have the state $\ket{\psi_f}$. Next, at Step (2), $OSM^*$ is applied on 
$\ket{\psi_f}$. The $OSM^*$ black box can then have two outputs:
\begin{itemize}
\item ``NO" - In that case $\ket{\psi_f}$ is not a product of $\ket{\pm}$ states. In particular, $\ket{\psi_f}\neq \pm\ket{+,...,+}$, implying that $f$ is not constant function. Consequently, $f$ is not a contradiction, implying that a satisfying assignment exists, and SAT answers ``YES"; similarly, $f$ is not a tautology (a function is a tautology when it is satisfied by all assignments). 
\item ``YES"  - In that case $\ket{\psi_f}$ is a product of $\ket{\pm}$ states and by Lemma \ref{lemma:2} we know that the Boolean function $f$ must be either constant or balanced. The protocol then continues to Step (3).
\end{itemize}

In Step (3) we are promised that the Boolean function $f$ is either constant or balanced. We can then use the Deutsch-Jozsa algorithm \cite{DeutschJozsa_1992} which differentiates between balanced and constant Boolean functions with just a single application of the circuit $U_f$. This has two possible outcomes:
\begin{itemize}
\item The function $f$ is balanced - then again $f$ is not constant, so SAT outputs ``YES" (and $f$ is not a tautology). 
\item The function $f$ is constant - then $f$ is either a tautology or a contradiction.  The protocol then continues to Step (4).
\end{itemize}
In Step (4) the protocol distinguishes between the possibility that $f$ is a tautology, in which case SAT  would output ``YES", and the option that $f$ is a contradiction, in which case SAT would output ``NO". This is done by a simple evaluation of $f$ on one fixed and arbitrary input using $U_f$. In both cases a definite output for the SAT problem (as well as for the tautology problem) is obtained, and the protocol  ends. \quad $\blacksquare$
\end{proof}

\section{Proofs for Section \ref{secsec:classical_complexity}} \label{app:classical_reduction}
\begin{lemma}
If a Boolean function $f$ is constant then the composed Boolean function $g=f(x_1,\ldots,x_n)\land x_{n+1}$ is either constant or balanced; Otherwise, i.e.\ if $f$ is not constant, then $g$ neither constant nor balanced.
\end{lemma}
\begin{proof}
If is enough to notice that: 
\begin{itemize}
\item The function $g$ is a contradiction if and only if $f$ is a contradiction.
\item The function $g$ cannot be a tautology.
\item The function $g$ is balanced if and only if $f$ is a tautology.
\end{itemize}
It follows directly that $g$ is constant or balanced if and only if $f$ is a constant. \quad $\blacksquare$
\end{proof}

\begin{theorem}
SAT is Karp-reducible to $\lnot cOSM^*$:
  \be
	{}^\lnot cOSM^* \geq_m^P SAT
  \ee
\end{theorem}
\begin{proof}
Given a Boolean function $f\!\!\!:\!\!\{0,1\}^{n} \!\rightarrow\! \{0,1\}$ for which we want to decide SAT, we (efficiently) construct another Boolean function $g\!\!\!:\!\!\{0,1\}^{n+2} \!\rightarrow\! \{0,1\}$, such that $g(x_1,\ldots,x_{n+2}) = f(x_1,\ldots,x_n) \land x_{n+1} \land x_{n+2}$, and show that $\ket{\psi_g}$ is not $OSM^*$ if and only if the function $f$ is satisfiable.

First, if $f$ is not satisfiable, i.e.\ a contradiction, then $g$ is also a contradiction and $\ket{\psi_g}$ is a constant sign state (defined in Section \ref{sec:complete_separability}). Thus when $f$ is not satisfiable, $\lnot cOSM^{*}$ outputs ``no" on $g$. Next, note that due to the particular structure of the function $g$, at most a quarter of all its possible assignments are satisfying assignments (where exactly quarter of satisfying assignments exist for a tautology function $f$, of which all assignments are satisfying assignments). This implies that for any  satisfiable function $f$, the number of satisfying assignments for $g$ is between one and $\frac{1}{4}2^{n+2}$, implying that $g$ is neither constant nor balanced and therefore  $\lnot cOSM^{*}$ for $g$ outputs: ``yes", as required. \quad $\blacksquare$

\end{proof}

%%%%%%%%%%%%%%%%%%%%%%%%%%%%%%%%%%%%%%%%%%%%%%%%%%%%%%%%%%%%%%%%%%%%%%%%%%%%%%%%%%%%%%%%%%%%%%%%%%%%%%%%%%%%%%%%%%%%%%%%%%%%%
%\input{Biblography}
\bibliography{bibliography}
\bibliographystyle{unsrt}%{plain}
\end{document}